\documentclass[letterpaper,conference]{ieeeconf}

\IEEEoverridecommandlockouts

\usepackage{url}
\usepackage{amsmath}
\usepackage{amssymb}
\usepackage{graphicx}
\usepackage{algorithm,algorithmic}	
\usepackage{amsmath}
\usepackage{tikz}
\usepackage{epstopdf}
\usepackage{mathtools}
\usepackage[compress]{cite}
\usepackage{flushend}

\newtheorem{problem}{Problem}
\newtheorem{proposition}{Proposition}

\newtheorem{remark}{Remark}

\newtheorem{assumption}{Assumption}
\newtheorem{theorem}{Theorem}
\newtheorem{corollary}{Corollary}

\newcommand{\argmin}{\mathop{\text{arg\,min}}}
\newcommand{\argmax}{\mathop{\text{arg\,max}}}

\newcommand{\trace}{\mathrm{Tr}}

\newcommand{\supp}{\mathrm{supp}}

\DeclareRobustCommand\approximation{\raisebox{1.3pt}{\tikz{\draw[-,red,solid,line width = 0.9pt](0,0) -- (5mm,0);}}}

\DeclareRobustCommand\ole{\raisebox{1.3pt}{\tikz{\draw[-,blue,solid,line width = 0.9pt](0,0) -- (5mm,0);}}}

\DeclareRobustCommand\map{\raisebox{1.3pt}{\tikz{\draw[-,black!30!green,solid,line width = 0.9pt](0,0) -- (5mm,0);}}}

\usepackage{tcolorbox}
\definecolor{mycolor}{rgb}{0.122, 0.435, 0.698}
\makeatletter
\newcommand{\mybox}[1]{%
  \setbox0=\hbox{#1}%
  \setlength{\@tempdima}{\dimexpr\wd0+13pt}%
  \begin{tcolorbox}[colframe=mycolor,boxrule=0.5pt,arc=4pt,
      left=6pt,right=6pt,top=6pt,bottom=6pt,boxsep=0pt,width=\@tempdima]
    #1
  \end{tcolorbox}
}
\makeatother

\usepackage{dsfont}

\renewcommand{\footnoterule}{
  \kern -2pt
  \hrule width 0.3\textwidth height .5pt
  \kern 2pt
}

\title{Security versus Privacy}

\author{
Farhad~Farokhi and Peyman Mohajerin Esfahani
\thanks{F. Farokhi is with the CSIRO's Data61 and the Department of Electrical and Electronic Engineering at the University of Melbourne, Australia. e-mail: ffarokhi@unimelb.edu.au, farhad.farokhi@data61.csiro.au}
\thanks{P. Mohajerin Esfahani is with the Delft Center for Systems and Control at the Delft University of Technology, the Netherlands. e-mail: P.MohajerinEsfahani@tudelft.nl}
\thanks{The work of F.~Farokhi was supported by the McKenzie Fellowship from the University of Melbourne, the VESKI Victoria Fellowship from the Victorian State Government, and a grant (MyIP: ID6874) from Defence Science and Technology Group (DSTG).}
\thanks{The work of P.~Mohajerin~Esfahani was supported by the Swiss National Science Foundation under the grant P2EZP2\_165264.}
}

\begin{document}
\maketitle

\begin{abstract} Linear queries can be submitted to a server containing private data. The server provides a response to the queries systematically corrupted using an additive noise to preserve the privacy of those whose data is stored on the server. The measure of privacy is inversely proportional to the trace of the Fisher information matrix. It is assumed that an adversary can inject a false bias to the responses. The measure of the security, capturing the ease of detecting the presence of the false data injection, is the sensitivity of the Kullback-Leiber divergence to the additive bias. An optimization problem for balancing privacy and security is proposed and subsequently solved. It is shown that the level of guaranteed privacy times the level of security equals a constant. Therefore, by increasing the level of privacy, the security guarantees can only be weakened and \textit{vice versa}. Similar results are developed under the differential privacy framework.
\end{abstract}

\section{Introduction}
Various frameworks, such as differential privacy~\cite{Dwork2011}, have been introduced to protect the privacy of individuals whose data is stored in online databases. These methods most often rely on the addition of noises with Laplace or Gaussian distributions to the outcome of queries on the databases containing the private information. More recently, differential privacy has found its way to control systems and signal processing~\cite{li2017privacy,le2014differentially,sandberg2015differentially,huang2014cost}. In addition to differential privacy, information theoretic methods (using mutual information or Fisher information as a measure of privacy) have been also developed within the control and estimation community for preserving the privacy of individuals~\cite{rajagopalan2011smart,farokhisandberg2016,kung2017compressive}. These methods also rely on the addition of noises which can be tailored for the specific problem at hand in order to protect the private data. 

Although privacy preserving, the additive noise might also make it harder for an outsider to be able to use the reported data for identifying malicious behavior. For instance, in the smart meter privacy examples in~\cite{rajagopalan2011smart,farokhisandberg2016}, a battery (which can be modeled as an additive noise with bounded support) is being used to mask the consumption patterns of the household. This ensures the privacy of the household. However, the battery operation will also makes it hard for the power authority to learn about the presence of malicious agents based on the provided smart meter data. This is because deviations of the smart meter readings from the power authority's expectations (formed on the basis of historical data or models of household consumption) can be attributed equally to the implemented privacy-preserving mechanism or a malicious entity. A systematic analysis of the trade-off between privacy and security is the topic of this paper.

Specifically, a problem setup is considered in which everyone can submit linear queries to an online server containing a vector of private data. The server, in return, provides a systematically corrupted response to the submitted queries. The corruption involves using an additive noise to preserve the privacy of the entries of the database, i.e., the aforementioned vector of private data. The server determines the statistics of the noise so that estimation error of the vector of private data is maximized under a constraint on the quality of the supplied response, captured by the variance of the additive noise. Noting that the estimation error of the private vector is a function of policy used for generating the estimate, the Cram\'{e}r-Rao bound~\cite[p.\,169]{cramerraotheorem}  is used to develop a universal measure of privacy which is inversely proportional to the trace of the Fisher information matrix. This measure of privacy is independent of the actions of the eavesdropper and is thus universal. It is assumed that an adversary can inject a bias to the server's response. The ability of users to be able to detect the presence of a bias (and thus raising a security alarm) is related to the Kullback-Leiber divergence of the output distribution with and without the additive bias. This provides a measure of security. The choice is motivated by the Chernoff-Stein Lemma (see, e.g.,~\cite{cover2012elements}) relating the probability of false negative (in the sense that a false bias injection attack escaping undetected) when using likelihood ratio hypothesis testing is a decreasing function of the Kullback-Leiber divergence of the output distribution with and without the additive bias. An optimization problem for balancing between privacy and security is proposed and solved. The solution in fact shows that the level of guaranteed privacy times the level of security is upper bounded by a constant. Therefore, by increasing the level of privacy, the security guarantees weaken and \textit{vice versa}. This observation can be generalized to any distribution in fact and is thus a fundamental property of the framework. Subsequently the differential privacy framework is studied for which the same limitation is also observed. 

Note that the use of Fisher information as a measure of privacy is not novel~\cite{anderson1977efficiency, farokhisandberg2016,farokhisandbergcdc2017}; however, a systematic method for balancing privacy and security is completely missing from the literature. This is the topic of the current paper.

Recently, in~\cite{giraldo2017security}, it was shown that differential privacy noise can prevent detection of integrity attack in dynamical systems. This is because the additive noise of differential privacy provides new avenues for an attacker to inject false information without raising suspicion. The results of this paper, although having similar interpretations, are different from~\cite{giraldo2017security}. Most importantly, using the Fisher information as a measure of privacy and the Kullback-Leiber divergence as a measure of security, we can develop a more fundamental understanding of the trade-off between security and privacy without restricting the framework to differential privacy. 

The rest of the paper is organized as follows. First, the problem formulation introducing the measures of privacy and security is presented in Section~\ref{sec:problem}. The results capturing the trade-off between privacy and security are then developed in Section~\ref{sec:results}. Finally, the paper is concluded in Section~\ref{sec:conclusions}.

\section{Problem Formulation} \label{sec:problem}
Consider the communication block diagram in Figure~\ref{fig:1}. A trustworthy server has access to a vector $x\in\mathcal{X}\subseteq\mathbb{R}^{n}$ whose entries must be kept private. Any agent, including those with an interest on infringing on the privacy of the individuals whose data is stored on the server, can submit a linear query of the form $Cx$ to the server  with observation matrix $C\in\mathbb{R}^{m\times n}$. 

\begin{assumption} $C$ has full row rank.
\end{assumption}

The server returns a response to the query of the form $z=Cx+w$, where $w\in\mathbb{R}^m$ is an additive privacy-preserving noise with probability density function $\gamma:\mathbb{R}^m\rightarrow\mathbb{R}_{\geq 0}$. 

\begin{assumption} \label{assum:2} $\gamma$ is twice continuously differentiable and $\supp(\gamma):=\{w\in\mathbb{R}^m\,|\,\gamma(w)>0\}$ may only differ from $\mathbb{R}^m$ over a Lebesgue measure zero set.
\end{assumption}

These are technical assumptions that allow us to efficiently capture the optimal trade-off between security and privacy. The first part of Assumption~\ref{assum:2} simplifies the search for the optimal privacy-preserving policy by allowing the use of tools available from the calculus of variations~\cite{kirk2004optimal}. The second part of Assumption~\ref{assum:2} ensures that the Fisher information matrix is well-defined and its trace is a convex function of $\gamma$~\cite{farokhisandbergcdc2017}. The set of all such probability density functions is denoted by $\Gamma$.

\subsection{Measure of Privacy}
In this paper, the Fisher information is utilized as a measure of privacy. In fact, the server aims at increasing 
\begin{align} \label{eqn:privacy_cost}
\mathcal{P}(\gamma):=1/\trace(W\mathcal{I}),
\end{align}
where the weighting matrix $W$ is a positive definite matrix and $\mathcal{I}$ is the Fisher information matrix defined as
\begin{align*}
\mathcal{I}:=\int \frac{\partial \log(\gamma(w))}{\partial w}\frac{\partial \log(\gamma(w))}{\partial w}^\top \gamma(w)\mathrm{d}w.
\end{align*}
Note that the Fisher information matrix is a function of the probability density function $\gamma$. This measure has been recently utilized within privacy literature; see, e.g,~\cite{farokhisandberg2016,farokhisandbergcdc2017}. The motivation behind this selection is given in what follows. 

The server wishes to keep the entries of the vector $x$ private. Therefore, it aims to select a probability density function $\gamma\in\Gamma$ to maximize $\mathbb{E}\{\|\Pi_x(x-\hat{x}(y))\|_2^2\}$, where $\Pi_x$ is a weighting matrix and $\hat{x}(y)$ is an estimator that an eavesdropper may use to estimate the value of the vector $x$ based on the received message $y$. 

Noting that the term $\mathbb{E}\{\|\Pi_x(x-\hat{x}(y))\|_2^2\}$ is a function of $\hat{x}(y)$, which makes the privacy measure depending on the eavesdropper (whose actions may not be known in advance), a lower bound of this term based on the Fisher information matrix is optimized. Using the Cram\'{e}r-Rao bound~\cite{533723}, under mild assumptions, it can be shown that
\begin{align*}
\mathbb{E}\{\|\Pi_x(x\hspace{-.03in}-\hspace{-.03in}\hat{x}(y))\|_2^2\}
=&\trace(\Pi_x^\top \Pi_x \mathbb{E}\{(x\hspace{-.03in}-\hspace{-.03in}\hat{x}(y))(x\hspace{-.03in}-\hspace{-.03in}\hat{x}(y))^\top\})\\
\geq &\trace(\Pi_x^\top \Pi_x ((g(x)-x)(g(x)-x)^\top \\
&\hspace{.7in}+G(x)\mathcal{I}_x^{\dag}G(x)^\top))\\
\geq &\trace(\Pi_x^\top \Pi_x (g(x)-x)(g(x)-x)^\top)\\
&+\trace(\Pi_x^\top \Pi_x\mathcal{I}_x^{\dag})\lambda_{\min}(G(x)^\top G(x))
\end{align*}
where $g(x)=\mathbb{E}\{\hat{x}(y)\}$, $G(x)$ is the Jacobian of $g(x)$, $\mathcal{I}_x=C^\top \mathcal{I}C$, and $X^\dag$ denotes the Moore-Penrose pseudo-inverse of any matrix $X$. Note that, if $G(x)$ is a full rank matrix (e.g., for all unbiased estimators), $\trace(\Pi_x^\top \Pi_x\mathcal{I}_x^{\dag})$ can be utilized as a measure of privacy that is independent of the behavior of the adversary. This is because by increasing $\trace(\Pi_x^\top \Pi_x\mathcal{I}_x^{\dag})$, the estimation error also increases. Noting that $\trace(\Pi_x^\top \Pi_x\mathcal{I}_x^{\dag})$  is not a concave function of $\gamma$, the measure of privacy can be replaced with $1/\trace((\Pi_x^\top \Pi_x)^{\dag}\mathcal{I}_x)$ because\footnote{Note that, for any non-zero semi-definite matrix $A$, it can be deduced that $\trace(A^\dag)\trace(A)\geq \trace(A^\dag A)\geq 1$ while implies that $\trace(A)\geq 1/\trace(A^\dag)$.} $\trace(\Pi_x^\top \Pi_x\mathcal{I}_x^{\dag})\geq 1/\trace((\Pi_x^\top \Pi_x)^{\dag}\mathcal{I}_x)$.
Interestingly, $1/\trace((\Pi_x^\top \Pi_x)^{\dag}\mathcal{I}_x)$ is a concave function of the probability density function $\gamma$ because $\trace((\Pi_x^\top \Pi_x)^{\dag}\mathcal{I}_x)$ is a convex function~\cite{farokhisandberg2016}. Thus maximizing $1/\trace((\Pi_x^\top \Pi_x)^{\dag}\mathcal{I}_x)$ is a more computationally-friendly task. Note that, for this motivational example, the weighting function in $\mathcal{P}$ is given by $W=C(\Pi_x^\top \Pi_x)^\dag C^\top$. 


\begin{remark}[Worst-case analysis] In the preceeding discussion, it is assumed that $G(x)$ is full rank, which although sensible (as estimators, such as least mean square, meet this condition), might not desirable. In~\cite{Farokhi_journal_2017}, it was shown that  $1/\trace(CC^\top \mathcal{I})$ can be proved to be a measure of privacy by studying worst-case privacy violations. In worst-case privacy attack, an eavesdropper has access to all the entries of the vector $x$ except one of them (the subject of the privacy infringement or eavesdropping attack) and it would like to infer the value of that entry based on the response to the submitted query. In that case, the weighting function in $\mathcal{P}$ is given by $W=C C^\top$. 
\end{remark} 

\subsection{Measure of Performance}
Noting that the error $\mathbb{E}\{\|\Pi_x(x-\hat{x}(y))\|_2^2\}$ can be made potentially unbounded (since there is no prior on $x$ and the server can add a Gaussian noise with increasing co-variance), the server also aims at maintaining a sensible level of performance by enforcing that
\begin{align}
\mathcal{Q}(\gamma):=\mathbb{E}\{\|y-Cx\|_2^2\}
\end{align}
remains below a certain level $\eta$, i.e., it is desired that the variance of the probability density function $\gamma$ is less than the provided upper bound by ensuring that
$$
\int w^\top w \gamma(w) \leq \eta.
$$
 
\begin{figure}[t]
\centering
\small
\begin{tikzpicture}
\node[draw,rectangle,minimum width=0.9cm,minimum height=0.7cm] (c) at (-0.5,0) {$C$};  
\node[] (x) at (-1.8,0) {$x$};
\draw[->] (x) -- (c);
\node[draw,circle] (s1) at (1.3,0) {};
\node[] at (1.3,0) {$+$};
\draw[->] (c) -- node[above] {$Cx$} (s1);
\node[] (w) at (1.3,0.8) {$w$};
\draw[->] (w) -- (s1);
\node[draw,circle] (s2) at (2.3,0) {};
\node[] at (2.3,0) {$+$};
\node[] (d) at (2.3,0.8) {$d$};
\draw[->] (d) -- (s2);
\draw[->] (s1) -- (s2);
\node[] (y) at (4.3,0) {$y=Cx+w+Fd$};
\draw[->] (s2) -- (y);
\end{tikzpicture}
\caption{\label{fig:1} Communication diagram.}
\end{figure}
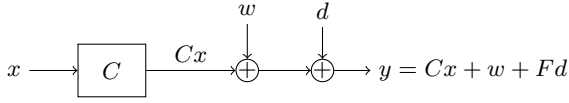


\subsection{Measure of Security}
The communication channel can be infiltrated by an adversary, which may inject the bias $d\in\mathbb{R}^p$. Thus, the final output is given by $y=Cx+w+Fd$. Therefore, the sever may also wish to make it possible for potential users to identify bias injection attacks performed by an adversary (to raise an alarm). This can be achieved by enforcing a constraint on an appropriately selected measure of security. 

In this paper, the sensitivity of Kullback-Leibler divergence between the probability density functions $\gamma(y-Cx)$ and $\gamma(y-Cx-Fd)$ is used as a measure of security. This shows that how easy it is for to distinguish between the probability density functions $\gamma(y-Cx)$ and $\gamma(y-Cx-Fd)$ for small offset term $d$ (which an adversary may use to avoid being detected). Therefore, the measure of security is given by 
\begin{align*}
\mathcal{S}(\gamma):=\min_{\xi\in \mathbb{R}^p}\lim_{d=\varrho \xi,
\varrho\rightarrow 0}\frac{\mathcal{KL}}{\|d\|_2^2},
\end{align*}
where 
\begin{align*}
\mathcal{KL}:=\int \gamma(y-Cx)\log\bigg(\frac{\gamma(y-Cx)}{\gamma(y-Cx-Fd)} \bigg)\mathrm{d}y.
\end{align*}
In this framework, it is desired to ensure that $\mathcal{S}(\gamma)\geq \alpha$, where $\alpha>0$ is an appropriately selected constant. 

This choice is motivated by that, at least for discrete random variables, it can be proved that the probability of false negative in the sense that a bias injection attack remains undetected when using likelihood ratio hypothesis testing is a decreasing function of $\mathcal{KL}$~\cite[Chernoff-Stein Lemma]{cover2012elements}. Although such a result may not hold in general, this observation can be used a motivation for the use of the Kullback-Leibler divergence as a measure of security. Thus, to keep the probability of false negatives small, a constraint of the form $\mathcal{KL}/\|d\|_2^2\geq \alpha$. Note that $\mathcal{KL}$ grows as a function of $d$ so a scaled version of the Kullback-Leibler divergence $\mathcal{KL}/\|d\|_2^2$ is considered. Considering that the adversary intends to not be detected (and the fact that the identification becomes easier as $\|d\|$ grows), it would be of interest to study small bias vectors $d$.

\subsection{Balancing Privacy and Security}
With the definitions of the measures of privacy and security in hand, it is now time to pose the problem mathematically.

\begin{problem} \label{prob:1} Find privacy preserving policy 
\begin{subequations}
\begin{align}
\gamma^*\in\argmax_{\gamma\in\Gamma} \quad & \mathcal{P}(\gamma),\\
 \mathrm{s.t.}\hspace{.1in}\quad & \mathcal{S}(\gamma)\geq \alpha, \\
&\mathcal{Q}(\gamma)\leq \eta.
\end{align}
\end{subequations}
\end{problem}


A popular framework for studying privacy is differential privacy; see, e.g.,~\cite{dwork2008differential}. The server's response is $\epsilon$-differentially private if 
\begin{align}
\mathbb{P}\{y\in\mathcal{Y}|x'\}\leq \exp(\epsilon)\mathbb{P}\{y\in\mathcal{Y}|x\}
\end{align}
for all $x,x'\in\mathcal{X}$ that are only different in maximum on entry and $\mathcal{Y}$ is a Lebesgue-measurable subset of $\mathbb{R}^m$. 

\begin{problem} \label{prob:2} Find $\epsilon$-differentially private $\gamma^*\in\Gamma$ such that $\mathcal{S}(\gamma)\geq \alpha$.
\end{problem}

Studying Problem~\ref{prob:2} allows us to see if we can observe the same results as in~\cite{giraldo2017security} within this setup. Furthermore, it can be investigated that if such results are in agreement with the optimal additive noise extracted from solving Problem~\ref{prob:1}.


\section{Main results} \label{sec:results}

The first result of this paper, formalized in the following theorem, states that the Gaussian additive noise with an appropriately selected co-variance matrix provides the best balance between privacy and security requirements according to Problem~\ref{prob:1}.

\begin{theorem} \label{tho:0}
The solution to Problem~\ref{prob:1} is given by
\begin{align*}
\gamma^*(w)=\frac{1}{\sqrt{\det(2\pi V_{ww})}}\exp\bigg(-\frac{1}{2}w^\top V_{ww}^{-1}w \bigg),
\end{align*}
where 
\begin{align*}
V_{ww}=\frac{\eta}{\trace(W^{1/2})}W^{1/2},
\end{align*}
if $\trace(W^{1/2})\lambda_{\min} (F^\top W^{-1/2}F)\geq 2\eta\alpha.$
\end{theorem}

\begin{proof} By eliminating the security constraint, Problem~\ref{prob:1} can be relaxed into
\begin{subequations}\label{eqn:2}
\begin{align} 
\gamma^*\in\argmin_{\gamma\in\Gamma}\quad&\trace(W\mathcal{I}),\\
\mathrm{s.t.}\hspace{.1in}\quad & \mathcal{Q}(\gamma)\leq \eta.
\end{align}
\end{subequations}
Note that the duality gap in~\eqref{eqn:2} is zero~\cite{jeyakumar1990zero}. Therefore, the constraint on the variance can be added to the cost function using a Lagrange multiplier, which transforms the problem into 
\begin{align} \label{eqn:1}
\max_{\lambda\geq 0}\min_{\gamma\in\Gamma}\trace(W\mathcal{I})+\lambda(\mathcal{Q}(\gamma)-\eta).
\end{align}
Following the same line of reasoning as in~\cite{farokhisandbergcdc2017}, the solution of the inner problem in~\eqref{eqn:1} is given by $\gamma^*(w)=u(w)^2$, where
\begin{align} \label{eqn:stationary}
\begin{cases}
\trace(W D^2 u(w))&\\
\hspace{.3in}+(\mu-(\lambda/4) w^\top w) u(w)=0, & w\in\mathcal{W},\\
u(w)=0, & w\in\partial\mathcal{W},\\
u(w)\neq 0, & w\in \mathrm{int}\mathcal{W},\\
\int_{w\in \mathcal{W}} u(w)^2\mathrm{d}w=1.
\end{cases}
\end{align}
Note that the cost function and the constraint set are convex, the stationarity condition in~\eqref{eqn:stationary} is sufficient for optimality. Further, if multiple density functions satisfy the conditions, they all exhibit the same cost. It can be shown that the following satisfies the stationarity condition:
\begin{align*}
u(w)=\frac{1}{\sqrt[4]{\det(2\pi V)}}\exp\bigg(-\frac{1}{4}w^\top V^{-1}w \bigg),
\end{align*}
where $V=W^{1/2}/\sqrt{\lambda}$. This shows that
\begin{align*}
\min_{\gamma\in\Gamma}&\;\trace(W\mathcal{I})+\lambda (\mathcal{Q}(\gamma)-\eta)=\trace(W V^{-1})+\lambda (\trace(V)-\eta).
\end{align*}
Therefore, the outer optimization problem in~\eqref{eqn:1} can be rewritten as
\begin{align*}
\max_{\lambda\geq 0}\, 2\trace(W^{1/2})\sqrt{\lambda}-\lambda\eta,
\end{align*}
and as a result $\lambda^*=\trace(W^{1/2})^2/\eta^2.$
%
Using~\cite{critchley1994preferred}, it can be shown that
\begin{align*}
\lim_{d=\varrho \xi,
\varrho\rightarrow 0}\frac{\mathcal{KL}}{\|d\|_2^2}=\frac{1}{2}\frac{\xi^\top \mathcal{I}_d \xi}{\xi^\top \xi}.
\end{align*}
where $\mathcal{I}_d:=F^\top\mathcal{I}F$.
Thus,
\begin{align*}
\mathcal{S}(\gamma)
&=\min_{\xi\in \mathbb{R}^p}
\lim_{d=\varrho \xi,\varrho\rightarrow 0} \frac{\mathcal{KL}}{\|d\|_2^2} =\frac{1}{2}\lambda_{\min}(\mathcal{I}_d)
\end{align*}
For $\gamma^*$, it can be seen that
\begin{align*}
\mathcal{I}_d=&\sqrt{\lambda^*}F^\top W^{-1/2}F=\trace(W^{1/2}) F^\top W^{-1/2}F/\eta.
\end{align*}
If $\mathcal{S}(\gamma)=(1/2)\lambda_{\min}(\mathcal{I}_d)\geq \alpha$, the solution of~\eqref{eqn:2} is also a solution of~\eqref{eqn:1}. This concludes the proof.
\end{proof}

Theorem~\ref{tho:0} presents the solution of Problem~\ref{prob:1} in the case where the constraint $\mathcal{Q}(\gamma)\leq \eta$ is active and $\mathcal{S}(\gamma)\geq \alpha$ is inactive. The following theorem extends this results to the case where the constraint $\mathcal{S}(\gamma)\geq \alpha$ is active and $\mathcal{Q}(\gamma)\leq \eta$ is inactive.


\begin{theorem} \label{tho:1}
Let 
\begin{subequations}\label{eqn:semidefinite}
\begin{align} 
\mathcal{V}:=\argmin_{X\succeq 0} & \quad \trace(W X),\\
\mathrm{s.t.} & \quad F^\top X F\succeq 2\alpha I.
\end{align}
\end{subequations}
The solution to Problem~\ref{prob:1} is given by
\begin{align}
\gamma(w)=\frac{1}{\sqrt{\det(2\pi V_{ww})}}\exp\bigg(-\frac{1}{2}w^\top V_{ww}^{-1}w \bigg)
\end{align}
if there exists $V_{ww}^{-1}\in\mathcal{V}$ such that $\trace(V_{ww})\leq \eta$.
\end{theorem}

\begin{proof} Note that 
$\mathcal{I}_d=F^\top \mathcal{I}F$. 
Assume that each $\mathcal{I}\succeq 0$ is realizable, i.e., there exists $\gamma(w)$ that results in it. Thus, Problem~\ref{prob:1} can be transformed into the semi-definite program in~\eqref{eqn:semidefinite}. It remains to find a density function that has a Fisher information equal to the solution of~\eqref{eqn:semidefinite}. This is in fact possible using a multivariate normal distribution with covariance matrix $\mathcal{I}^{-1}$. This concludes the proof.
\end{proof}

For scalar queries, such as averaging, the solution to Problem~\ref{prob:1} can be greatly simplified. This is shown in the following corollary.

\begin{corollary} \label{cor:1} For scalar queries (i.e., $m=1$), the solution to Problem~\ref{prob:1} is given by
\begin{align}
\gamma(w)=
\dfrac{1}{\sqrt{2\pi V_{ww}}}\exp\bigg(-\dfrac{w^2}{2V_{ww}} \bigg),
\end{align}
where 
\begin{align*}
V_{ww}=
\begin{cases}
\eta, & \eta\leq\lambda_{\min} (F^\top F)/\alpha,\\
1/\alpha, & \mbox{otherwise}.
\end{cases}
\end{align*}
\end{corollary}

\begin{proof} If $\eta\leq  \lambda_{\min} (F^\top F)/\alpha$, the results of Theorem~\ref{tho:0} can be used. Otherwise, the results of Theorem~\ref{tho:1} should be utilized in which case it can be seen that $\trace(W X)=WX$ (since both $X$ and $W$ are scalars) and $F^\top X F=(F^\top F) X $ (again because $X$ is a scalar). Hence, the optimization problem in~\eqref{eqn:semidefinite} can be transformed into $V_{ww}^{-1}\in\argmin_{X\geq \alpha} X$. Thus, $V_{ww}=1/\alpha$. This concludes the proof.
\end{proof}


\begin{corollary}\label{prop:1} For the optimal probability density function in Corollary~\ref{cor:1}, 
$\mathcal{S}(\gamma)\mathcal{P}(\gamma)=\lambda_{\min}(F^\top F)/(2W).$
\end{corollary}

\begin{proof}
For the optimal policy in Corollary~\ref{cor:1}, it can be seen that $\mathcal{KL}=\frac{1}{2} (F d)^2 V_{ww}^{-1},$ and, as a result, $\mathcal{S}(\gamma)=\frac{1}{2}\lambda_{\min}(F^\top F)/V_{ww}$. On the other hand, $\mathcal{P}(\gamma)=V_{ww}/W$.
\end{proof}

Proposition~\ref{prop:1} shows that by increasing $\mathcal{P}(\gamma)$ to achieve a higher privacy guarantee, $\mathcal{S}(\gamma)$ decreases, which makes the system more vulnerable to bias injection attacks. In fact, in lay terms, it can be expressed that
\begin{center}
\mybox{$\mbox{``privacy}\times \mbox{security} = \mbox{constant''}.$}
\vspace{-.3in}\hfill ($\star$)
\end{center}

In what follows, it is shown that Corollary~\ref{prop:1} and its interpretation in ($\star$) hold for any probability density function $\gamma(w)$ if $m=1$ (and not necessarily the solution of Problem~\ref{prob:1})

\begin{proposition}[Trade-off between Privacy and Security] \label{prop:2} For $m=1$, 
$\mathcal{S}(\gamma)\mathcal{P}(\gamma)=\lambda_{\min}(F^\top F)/(2W)$ for any $\gamma\in\Gamma$.
\end{proposition}

\begin{proof} For any density function, it can be seen that
\begin{align*}
\mathcal{S}(\gamma)=\lim_{d=\varrho \xi,
\varrho\rightarrow 0}\frac{\mathcal{KL}}{\|d\|_2^2}=\frac{1}{2}\frac{\xi^\top F^\top\mathcal{I}F \xi}{\xi^\top \xi}=\frac{1}{2}\lambda_{\min}(F^\top\mathcal{I}F).
\end{align*}
Thus,$\mathcal{P}(\gamma)\mathcal{S}(\gamma)=\lambda_{\min}(F^\top\mathcal{I}F)/({2\trace(W\mathcal{I})}).$ For $m=1$, it can be shown that $\mathcal{P}(\gamma)\mathcal{S}(\gamma)=\lambda_{\min}(F^\top F)/(2W)$ because $\mathcal{I}$ is scalar.
\end{proof}

\begin{figure}
\centering
\begin{tikzpicture}
\node[] at (0,0) {\includegraphics[width=1.0\columnwidth]{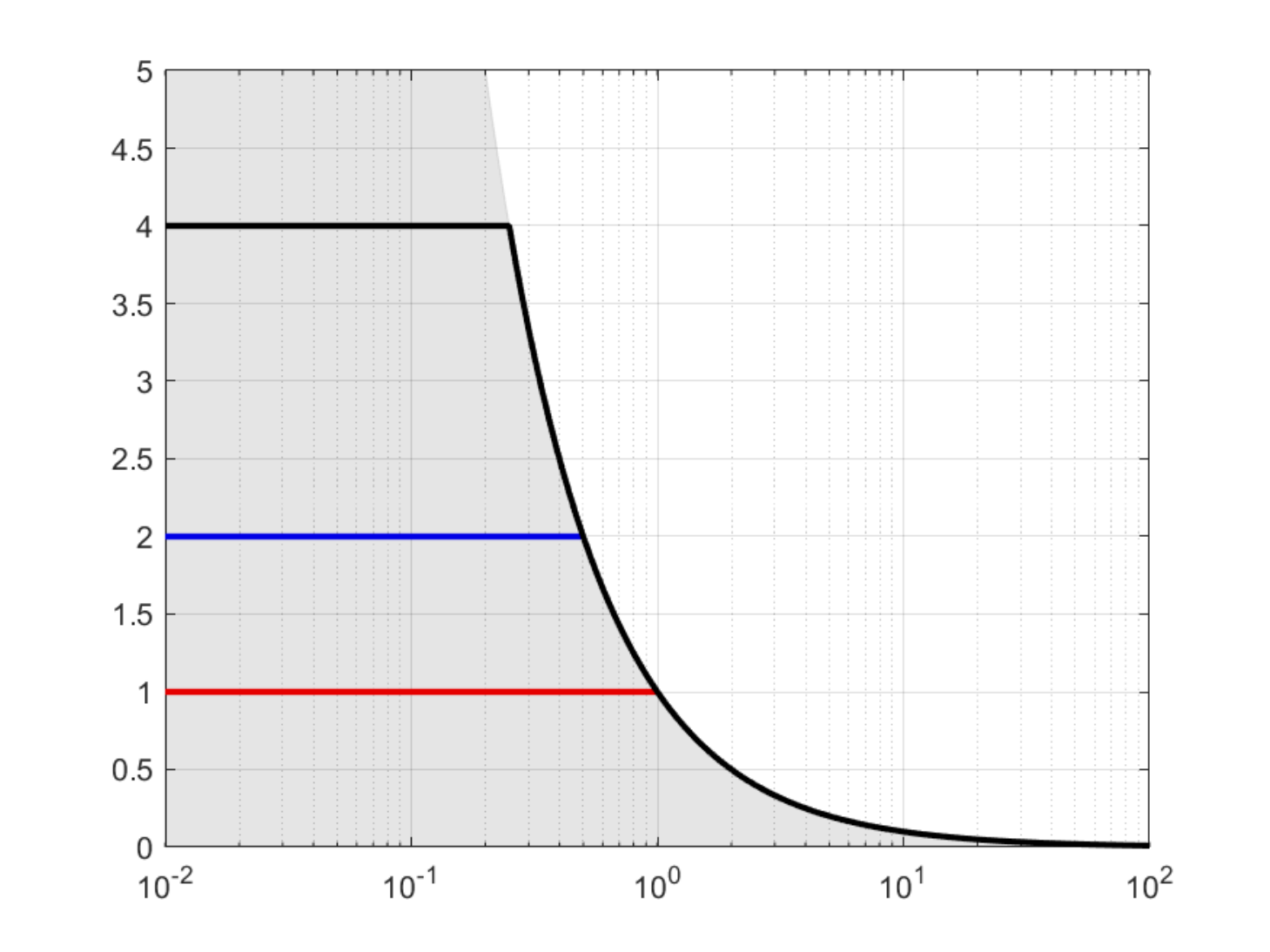}};
\node[rotate=90] at (-3.9,0) {privacy $\mathcal{P}(\gamma)$};
\node[] at (0,-3.2) {security guarantee $\alpha$};
\end{tikzpicture}
\caption{\label{fig:1} The trade-off between measure of privacy $\mathcal{P}(\gamma)$ for the optimal policy in Corollary~\ref{cor:1} versus the lower bound on the measure of security $\alpha$ for various response quality guarantees $\eta=1$ (solid red\,\approximation), $\eta=2$ (solid red\,\ole), and $\eta=4$ (solid green\,\map). The plateau on the achievable privacy guarantee for small $\alpha$ is caused by the constraint on the quality of measurement $\mathcal{Q}(\gamma)$. The gray area denotes the cases for which $\mathcal{P}(\gamma)\alpha\leq 1$. }
\end{figure} 

Figure~\ref{fig:1} illustrates the trade-off between measure of privacy $\mathcal{P}(\gamma)$ for the optimal policy in Corollary~\ref{cor:1} versus the lower bound on the measure of security $\alpha$ for various quality of response guarantees $\eta$. In this numerical example, $m=1$, $F=1$, and $W=1$. The plateau on the achievable privacy guarantee for small values of $\alpha$ is caused by the constraint on the quality of measurement $\mathcal{Q}(\gamma)$. The gray area denotes the cases for which $\mathcal{P}(\gamma)\alpha\leq 1$. All these cases are achievable for various values of $\eta$. Note that this not in contrast with the results of  Propositions~\ref{prop:1} and~\ref{prop:2} as they explore $\mathcal{P}(\gamma)\mathcal{S}(\gamma)$ (which is not necessarily equal to $\mathcal{P}(\gamma)\alpha$ as the constraint $\mathcal{S}(\gamma)\geq \alpha$ is not always active).

Now, we are ready to explore the solution of Problem~\ref{prob:2} regarding the balance between privacy and security in the differential privacy framework.

\begin{theorem} \label{tho:2} For scalar problems, i.e., $m=1$, and $\epsilon\geq \Delta\sqrt{2\alpha}$, the solution to Problem~\ref{prob:2} is given by
\begin{align}
\gamma(w)=\frac{1}{2\Delta/\epsilon}\exp\bigg(-\frac{|w|}{\Delta/\epsilon} \bigg),
\end{align}
where $\Delta:=\sup_{x,x'\in\mathcal{X}:\|x-x'\|_0\leq 1}|C(x-x')|$.
\end{theorem}

\begin{proof} Note that
\begin{align}
\frac{p(y|x')}{p(y|x)}&=\exp\bigg(\frac{|y-Cx|-|y-Cx'|}{\Delta/\epsilon}\bigg)
\nonumber\\
&\leq \exp\bigg(\frac{|C(x'-x)|}{\Delta/\epsilon}\bigg)\nonumber\\
&\leq \exp(\epsilon),\label{eqn:diffprivacy}
\end{align}
where the first inequality follows from that
$|y-Cx|=|y-Cx+Cx'-Cx'|\leq |y-Cx'|+|C(x'-x)|.$
Integrating both sides of~\eqref{eqn:diffprivacy} concludes the proof. Furthermore, $\gamma$ meets $\mathcal{I}_d=\epsilon^2/\Delta^2$. Thus, $\mathcal{I}_d\geq 2\alpha$ if and only if $\epsilon\geq \Delta\sqrt{2\alpha}$.
\end{proof}

For the $\epsilon$-differentiallay private distribution in Theorem~\ref{tho:2}, the following can be proved:
\begin{align*}
\mathcal{KL}=&\int \frac{1}{2\Delta/\epsilon}\exp\bigg(-\frac{|y-Cx|}{\Delta/\epsilon} \bigg)\\
&\hspace{.2in}\times\bigg(\frac{|y-Cx-Fd|}{\Delta/\epsilon} -\frac{|y-Cx|}{\Delta/\epsilon} \bigg)\mathrm{d}y\\
=&\int \frac{1}{2\Delta/\epsilon}\exp\bigg(-\frac{|\bar{y}|}{\Delta/\epsilon} \bigg)\bigg(\frac{|\bar{y}-Fd|}{\Delta/\epsilon} -\frac{|\bar{y}|}{\Delta/\epsilon} \bigg)\mathrm{d}\bar{y}\\
=&\exp(-|Fd|\epsilon/\Delta)-1+|Fd|\epsilon/\Delta.
\end{align*}
Therefore
\begin{align*}
\mathcal{S}(\alpha)
=\min_{\xi}\lim_{d=\varrho \xi,
\varrho\rightarrow 0}\frac{\mathcal{KL}}{\|d\|_2^2}
=\frac{\lambda_{\min}(F^\top F)\epsilon^2}{2\Delta^2}.
\end{align*}
This implies that by increasing the privacy guarantee (which is inversely proportional to $\epsilon$), the security level decreases and \textit{vice versa}. This is a similar observation to that of~($\star$). 

Figure~\ref{fig:3} illustrates the trade-off between measure of privacy $1/\epsilon$ versus the measure of security $\mathcal{S}(\gamma)$ for  the differentially-private policy in Theorem~\ref{tho:2}. Here, $m=1$, $\Delta=1$, and $F=1$. Recalling that $\mathcal{S}(\gamma)$ is motivated by small biases $d$, we also explore $\mathcal{KL}$ for differentially-private policies. This relationship is shown in Figure~\ref{fig:2}. Clearly, the same trend regarding the inverse relationship of the privacy and security can still be observed.

\begin{figure}
\centering
\begin{tikzpicture}
\node[] at (0,0) {\includegraphics[width=1.0\columnwidth]{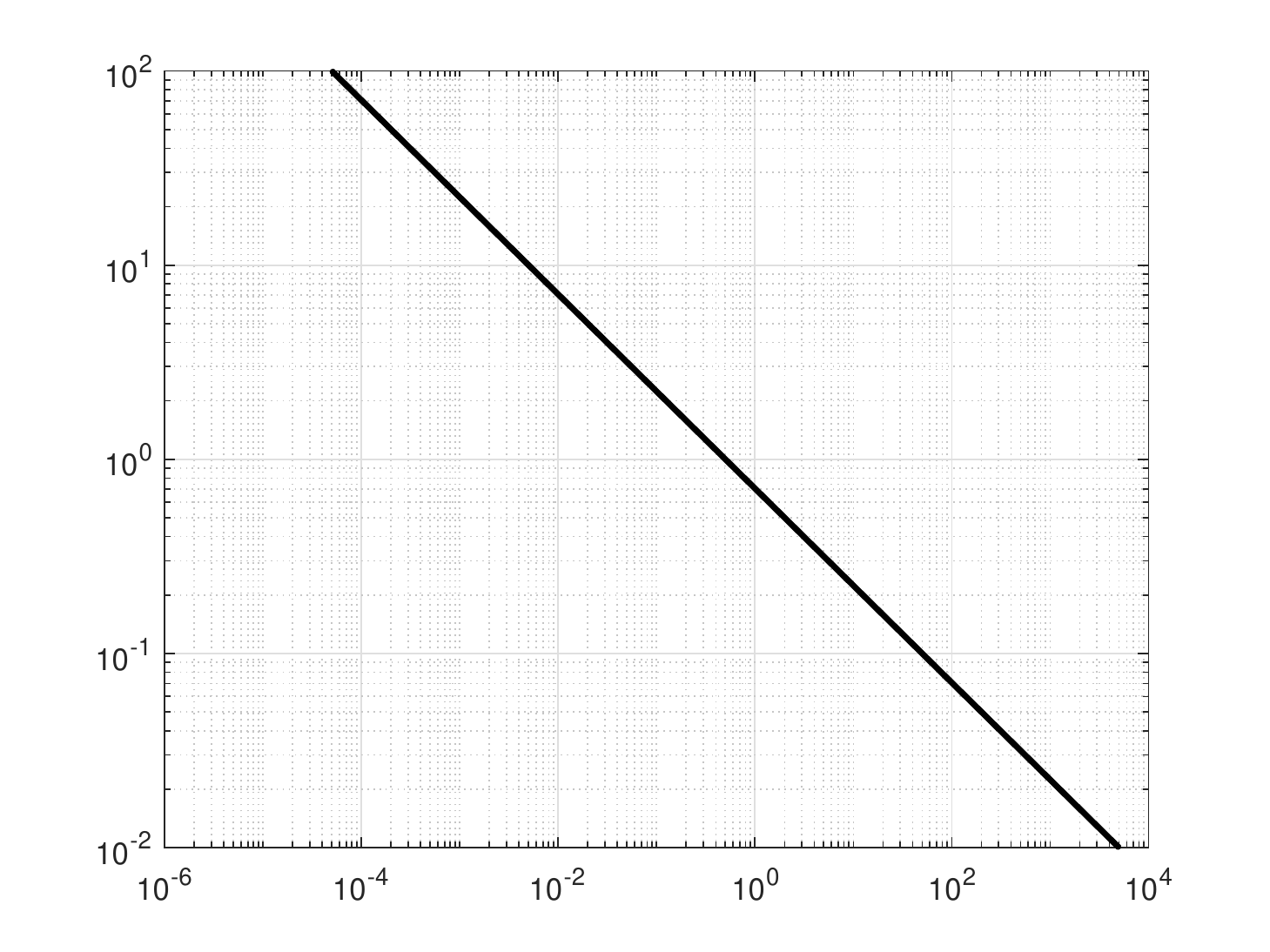}};
\node[rotate=90] at (-3.9,0) {privacy $1/\epsilon$};
\node[] at (0,-3.2) {security $\mathcal{S}(\alpha)$};
\end{tikzpicture}
\caption{\label{fig:3} The trade-off between measure of privacy $1/\epsilon$ and the measure of security $\mathcal{S}(\gamma)$ for  the differentially-private policy in Theorem~\ref{tho:2}. }
\vspace{-.2in}
\end{figure}

\begin{figure}
\centering
\begin{tikzpicture}
\node[] at (0,0) {\includegraphics[width=1.0\columnwidth]{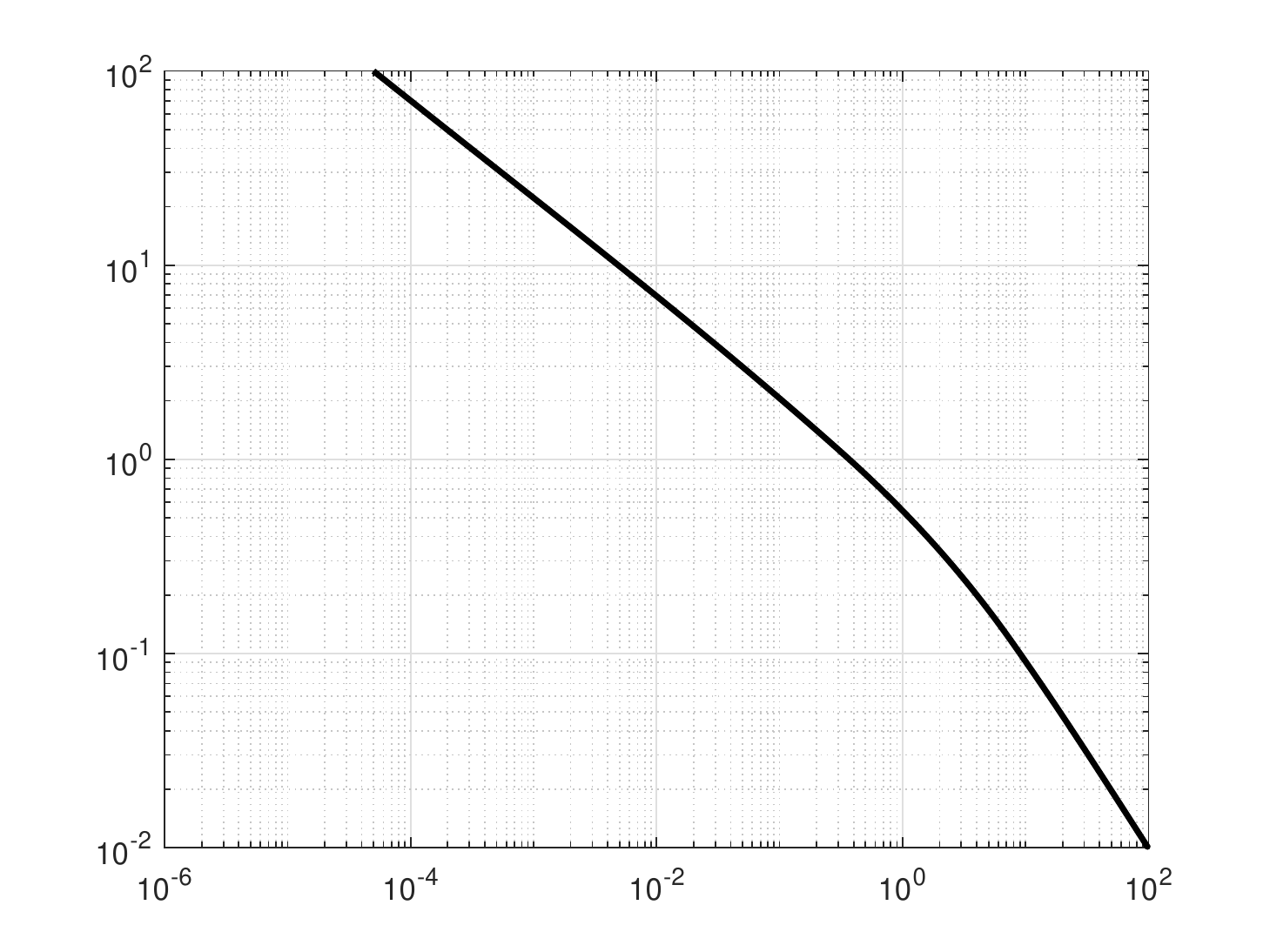}};
\node[rotate=90] at (-3.9,0) {privacy $1/\epsilon$};
\node[] at (0,-3.2) {security $\mathcal{KL}$};
\end{tikzpicture}
\caption{\label{fig:2} The trade-off between measure of privacy $1/\epsilon$ and $\mathcal{KL}$ for the differentially-private policy in Theorem~\ref{tho:2}. }
\end{figure} 
 
\section{Conclusions and Future Work} \label{sec:conclusions}
A framework was developed in which linear query can be submitted to a server containing private data. The server provides a response to the query corrupted using an additive noise to preserve the privacy of those whose data is on the server. It is shown that the level of guaranteed privacy times the level of security is always upper bounded by a constant and, as a result, higher privacy guarantees dictates weakened security guarantees. Future work can focus on dynamic problems.

\bibliographystyle{ieeetr}
\bibliography{ref1}

\begin{thebibliography}{10}

\bibitem{Dwork2011}
C.~Dwork, ``Differential privacy,'' in {\em Encyclopedia of Cryptography and
  Security} (H.~C.~A. van Tilborg and S.~Jajodia, eds.), Boston, MA: Springer
  US, 2011.

\bibitem{li2017privacy}
Z.~Li and T.~J. Oechtering, ``Privacy-constrained parallel distributed
  neyman-pearson test,'' {\em IEEE Transactions on Signal and Information
  Processing over Networks}, vol.~3, no.~1, pp.~77--90, 2017.

\bibitem{le2014differentially}
J.~Le~Ny and G.~J. Pappas, ``Differentially private filtering,'' {\em IEEE
  Transactions on Automatic Control}, vol.~59, no.~2, pp.~341--354, 2014.

\bibitem{sandberg2015differentially}
H.~Sandberg, G.~D\'{a}n, and R.~Thobaben, ``Differentially private state
  estimation in distribution networks with smart meters,'' in {\em Proceedings
  of the 54th IEEE Conference on Decision and Control}, pp.~4492--4498, 2015.

\bibitem{huang2014cost}
Z.~Huang, Y.~Wang, S.~Mitra, and G.~E. Dullerud, ``On the cost of differential
  privacy in distributed control systems,'' in {\em Proceedings of the 3rd
  International Conference on High Confidence Networked Systems}, pp.~105--114,
  ACM, 2014.

\bibitem{rajagopalan2011smart}
S.~R. Rajagopalan, L.~Sankar, S.~Mohajer, and H.~V. Poor, ``Smart meter
  privacy: A utility-privacy framework,'' in {\em Proceedings of the IEEE
  International Conference on Smart Grid Communications (SmartGridComm)},
  pp.~190--195, IEEE, 2011.

\bibitem{farokhisandberg2016}
F.~Farokhi and H.~Sandberg, ``Fisher information as a measure of privacy:
  Preserving privacy of households with smart meters using batteries,'' {\em
  IEEE Transactions on Smart Grid}, 2017.
\newblock In Press.

\bibitem{kung2017compressive}
S.-Y. Kung, ``Compressive privacy: From information/estimation theory to
  machine learning [lecture notes],'' {\em IEEE Signal Processing Magazine},
  vol.~34, no.~1, pp.~94--112, 2017.

\bibitem{cramerraotheorem}
J.~Shao, {\em Mathematical Statistics}.
\newblock Springer Texts in Statistics, Springer-Verlag New York, 2003.

\bibitem{cover2012elements}
T.~M. Cover and J.~A. Thomas, {\em Elements of Information Theory}.
\newblock Wiley, 2012.

\bibitem{anderson1977efficiency}
H.~Anderson, ``Efficiency versus protection in a general randomized response
  model,'' {\em Scandinavian Journal of Statistics}, pp.~11--19, 1977.

\bibitem{farokhisandbergcdc2017}
F.~Farokhi and H.~Sandberg, ``Optimal constrained additive noise distribution
  minimizing {Fisher} information for ensuring privacy,'' in {\em Proceedings
  of the 56th IEEE Conference on Decision and Control}, pp.~2692--2697, 2017.

\bibitem{giraldo2017security}
J.~Giraldo, A.~A. Cardenas, and M.~Kantarcioglu, ``Security vs. privacy: How
  integrity attacks can be masked by the noise of differential privacy,'' in
  {\em Proceedings of the American Control Conference}, pp.~1679--1684, 2017.

\bibitem{kirk2004optimal}
D.~E. Kirk, {\em Optimal Control Theory: An Introduction}.
\newblock Dover Books on Electrical Engineering Series, Dover Publications,
  2004.

\bibitem{533723}
A.~O. Hero, J.~A. Fessler, and M.~Usman, ``Exploring estimator bias-variance
  tradeoffs using the uniform {CR} bound,'' {\em IEEE Transactions on Signal
  Processing}, vol.~44, no.~8, pp.~2026--2041, 1996.

\bibitem{Farokhi_journal_2017}
F.~Farokhi and H.~Sandberg, ``Ensuring privacy with constrained additive noise
  by minimizing fisher information.''
\newblock Submitted, 2017.

\bibitem{dwork2008differential}
C.~Dwork, ``Differential privacy: A survey of results,'' in {\em Theory and
  Applications of Models of Computation: 5th International Conference, TAMC
  2008, Xi'an, China, April 25-29, 2008. Proceedings} (M.~Agrawal, D.~Du,
  Z.~Duan, and A.~Li, eds.), pp.~1--19, Berlin, Heidelberg: Springer Berlin
  Heidelberg, 2008.

\bibitem{jeyakumar1990zero}
V.~Jeyakumar and H.~Wolkowicz, ``Zero duality gaps in infinite-dimensional
  programming,'' {\em Journal of Optimization Theory and Applications},
  vol.~67, no.~1, pp.~87--108, 1990.

\bibitem{critchley1994preferred}
F.~Critchley, P.~Marriott, and M.~Salmon, ``Preferred point geometry and the
  local differential geometry of the {Kullback-Leibler} divergence,'' {\em The
  Annals of Statistics}, pp.~1587--1602, 1994.

\end{thebibliography}

\end{document}